\documentclass[a4paper]{article}
\pdfoutput=1
 
 \usepackage[margin=1.25in]{geometry}
\usepackage{verbatim}
\usepackage{graphicx}
\usepackage{amsbsy}
\usepackage{authblk}
\usepackage{amsmath,amssymb,color}
\usepackage{graphicx}
\usepackage{amsmath,amsfonts}
\usepackage{enumerate}
\usepackage{enumitem}
\usepackage{amsthm}
 \usepackage[normalem]{ulem}
 
\newtheorem{theorem}{Theorem}
\newtheorem{cor}{Corollary}

\newtheorem{lemma}{Lemma}

\renewcommand{\paragraph}[1]{\medskip\noindent\textbf{#1.}}

\newcommand{\changed}[1]{{\color{black} #1}}

\newcommand{\newnote}[1]{{\color{cyan} #1}}
\newcommand{\remove}[1]{{}}

\newcommand{\AMG}{angle-monotone graph }

\title{Angle-Monotone Graphs: Construction and Local Routing\footnote{This work is partially supported by NSERC.}}

\author{Anna Lubiw} 
\affil[1]{Cheriton School of Computer Science,  University of Waterloo,  Canada\\
  \texttt{alubiw@uwaterloo.ca}} 
\author{Debajyoti Mondal}
\affil[2]{Department of Computer Science, University of Saskatchewan, 
   Canada\\
  \texttt{d.mondal@usask.ca}}

\begin{document}

\maketitle

\begin{abstract}
A geometric graph in the plane is \emph{angle-monotone of width $\gamma$} if every pair of vertices is connected by an \emph{angle-monotone path of width~$\gamma$}, a path such that the angles of any two edges in the path differ by at most $\gamma$. Angle-monotone graphs have good spanning properties.

We prove that every point set in the plane admits an angle-monotone graph of width $90^\circ$, hence with spanning ratio $\sqrt 2$, and  a subquadratic number of edges.  
This answers an open question posed by Dehkordi, Frati and Gudmundsson. 

We show how to construct, for any point set of size $n$ and any angle $\alpha$, $0 < \alpha < 45^\circ$, an angle-monotone graph of width $(90^\circ+\alpha)$ with $O(\frac{n}{\alpha})$ edges.  
Furthermore, we give a local routing algorithm to find angle-monotone paths of width $(90^\circ+\alpha)$ in these graphs.
 The \emph{routing ratio}, which is the ratio of path length to Euclidean distance, is at most $1/\cos(45^\circ + \frac{\alpha}{2})$,
i.e.,~ranging from $\sqrt 2 \approx 1.414$ to $2.613$.    
 For the special case  $\alpha = 30^\circ$,  we obtain
 the $\Theta_6$-graph and our routing algorithm achieves the known routing ratio 2 while finding angle-monotone paths of width $120^\circ$.
\end{abstract}

\section{Introduction}
\label{sec:introduction}
The problem of constructing a geometric graph on a given set of points in  the plane
 so that the graph is sparse yet has good spanning and/or routing properties has been very well-studied. 
The basic goal is to guarantee paths that are relatively short, and to be able to find such paths using local routing. 
Two fundamental concepts in this regard are \emph{spanners} and \emph{greedy graphs}. A geometric graph 
is a \emph{$t$-spanner} if there is a path of stretch factor $t$ between any two vertices, i.e., a path whose length is at most $t$ times the Euclidean distance between the endpoints.  
A geometric graph 
 is \emph{greedy} if there is a path between every two vertices such that each intermediate vertex is closer to the destination than the previous 
 vertex on the path.  
 Greedy graphs permit \emph{greedy routing} where a path from source to destination is found by the local rule of moving from the current vertex to any neighbor that is 
 closer to the destination. 
However, greedy graphs
are not necessarily $t$-spanners for any constant $t$.

The most desirable goal would be to construct sparse geometric graphs together with a local routing algorithm to find paths with bounded  stretch factor that always get closer to the destination.
This is the topic of our paper. 
There are two aspects to the goal: to construct sparse geometric graphs in which such paths exist, and to find the paths via a local routing algorithm.   


Recently, Dehkordi et al.~\cite{D-Frati-G} introduced a class of graphs with good path properties: 
A graph is \emph{angle-monotone} if there is a path between every two vertices that, after some rotation, is $x$- and $y$-monotone---equivalently, there is some $90^\circ$ wedge such that the vector of every edge of the path lies in this wedge.  
This class was explored (and named) by Bonichon et al.~\cite{angle-mono}.
Any angle-monotone path $\sigma$ from $s$ to $t$  has the \emph{self-approaching} property
(see~\cite{self-approach}) that a point moving along $\sigma$ always gets closer to $t$. 

The concept can be generalized to wedges of angles other than $90^\circ$---a path is \emph{angle-monotone of width $\gamma$} (``generalized angle-monotone'') if there is some wedge of angle $\gamma$ such that the vector of every edge of the path lies in this wedge.  Although graphs that are angle monotone of width greater than $90^\circ$ are not necessarily self-approaching, they have good spanning properties. A graph that is angle-monotone of width $\gamma < 180^\circ$ is a $(1/\cos\frac{\gamma}{2})$-spanner~\cite{angle-mono}, thus a $\sqrt 2$ spanner for $\gamma=90^\circ$ (the factor $\sqrt 2$ is obvious based on the path being $x$- and $y$-monotone after some rotation). 

Our specific goal in this paper is to construct sparse generalized angle-monotone graphs and design local routing algorithms to find generalized angle-monotone paths in them. 
%
There have been a few results on constructing angle-monotone graphs, but no previous results on local routing to find angle-monotone paths---except for some impossibility results.


\paragraph{Constructing Angle-Monotone Graphs} The best result on constructing planar angle-monotone graphs is due to Dehkordi et al.~\cite{D-Frati-G} who proved that any set of $n$ points has a planar angle-monotone graph of width $90^\circ$ using $O(n)$ Steiner points. They proved this by showing that a
Gabriel triangulation is angle-monotone of width $90^\circ$ (see~\cite{Lubiw-O'Rourke} for a simpler proof), and then 
using the result that any point set can be augmented with $O(n)$ Steiner points to obtain a point set whose Delaunay triangulation is Gabriel. 
Without Steiner points, it is known that one cannot guarantee planar angle-monotone graphs for all point sets~\cite{angle-mono}. 
For the special case of $n$ points in convex position, 
Dehkordi et al.~\cite{D-Frati-G} proved that there exists a (non-planar) angle-monotone graph with $O(n\log n)$ edges. In this paper we show that any point set has an angle-monotone graph with a subquadratic number of edges.

Turning to angle-monotone graphs of larger width, Bonichon et al.~\cite{angle-mono} showed that the half-$\Theta_6$-graph  
on a set of $n$ points, which is planar, is an angle-monotone graph of width $120^\circ$.  

\paragraph{Local Routing on Angle-Monotone Graphs}
 A \emph{$k$-local routing algorithm} finds a path one vertex at a time using only local information about the current vertex and its $k$-neighborhood 
 plus the coordinates of the destination.
The \emph{routing ratio} of a local routing algorithm is the maximum stretch factor of any path found by the algorithm. 
%
%
The results mentioned in the previous two paragraphs imply that Gabriel graphs are $\sqrt 2$-spanners, and half-$\Theta_6$-graphs are $2$-spanners (as was previously known~\cite{bonichonTD,Chew86}). Are there local routing algorithms to find paths with good 
 stretch factors,
or paths that are angle-monotone in these classes of graphs? 
The answers are ``yes'' and ``no'', respectively.
Bonichon et al.~\cite{angle-mono} gave a 1-local routing algorithm for Gabriel graphs that 
has routing ratio $(1 + \sqrt 2)$.  
On the other hand, they proved that no local routing algorithm can find angle-monotone paths in Gabriel graphs.
Bose et al.~\cite{BoseTheta6journal} gave a 1-local routing algorithm for half-$\Theta_6$-graphs that 
has routing ratio $2.887$.
They proved that this is the best ratio possible for any local routing algorithm, which implies that no local routing algorithm will find angle-monotone paths of width $120^\circ$ in half-$\Theta_6$-graphs. 
We construct a family of graphs together with a local routing algorithm that finds generalized angle-monotone paths.
 

\paragraph{Contributions} Our main results are as follows: 

{\bf 1.} Given $n$ points in the plane we construct an angle-monotone graph of width $90^\circ$ with  $O(\frac{n^2\log\log n}{\log n})$ edges---a subquadratic number of edges. Since angle-monotone graphs are \emph{increasing-chord graphs}, this answers Open Problem 4 from~\cite{D-Frati-G}. 
 (We refer to~\cite{self-approach,BahooDMM17,MastakasS15,NollenburgPR16,Rote}
for results on self-approaching and increasing-chord graphs.) 

{\bf 2.} Given $n$ points in the plane and any $\alpha$, $0 < \alpha <  45^\circ$, we construct an angle-monotone graph of width $90^\circ + \alpha$ with $O(\frac{n}{\alpha})$ edges.    We give a 2-local routing algorithm for these graphs that finds angle-monotone paths of width $90^\circ + \alpha$, 
 thus of stretch factor $1/\cos(\frac{90^\circ + \alpha}{2})$. 
In particular, for $\alpha = 30^\circ$ our construction yields the [full] $\Theta_6$-graph, and our local routing algorithm finds angle-monotone paths of width $120^\circ$ and stretch factor $2$. 
 For this case, our algorithm is very similar to the one of Bose et al.~\cite{BoseTheta6journal} that finds paths of stretch factor 2 in half-$\Theta_6$-graphs, 
but our proof of correctness is simpler.

\remove{
\begin{compactenum}
\item Given $n$ points in the plane we construct an angle-monotone graph of width $90^\circ$ with a subquadratic number of edges. Since angle-monotone graphs are \emph{increasing-chord graphs}, this answers Open Problem 4 from~\cite{D-Frati-G}. \marginnote{missing the defn of `increasing chord'.  I put in italics.  I'm ok with no definition.}

\item Given $n$ points in the plane and any $\alpha$, $0 < \alpha <  45^\circ$, we construct an angle-monotone graph of width $90^\circ + \alpha$ with $O(\frac{n}{\alpha})$ edges.    We give a local routing algorithm for these graphs that finds angle-monotone paths of width $90^\circ + \alpha$, 
\newnote{thus of stretch factor $1/\cos(\frac{90^\circ + \alpha}{2})$.} 
In particular, for $\alpha = 30^\circ$ our construction yields the [full] $\Theta_6$-graph, and our local routing algorithm finds angle-monotone paths of width $120^\circ$ \newnote{and stretch factor $2$.} 
Achieving routing ratio $2$ was known~\cite{BoseTheta6journal} but finding angle-monotone paths is new.  Furthermore, our routing algorithm and the proof of its correctness are simpler than those of~\cite{BoseTheta6journal}.
\end{compactenum}
}

%
%

\remove{
\section{Preliminaries}

We introduced various properties of paths and of graphs above, e.g.~greedy, self-approaching, angle-monotone.  
In general, we use the terminology that a graph is defined to be of type X if for every two vertices $s$ and $t$ there is a path of type X from $s$ to $t$.  This pattern can be used to fill in the definitions that were not made formal above.  


 
A routing algorithm is \emph{local} if it uses only local information available at the current node to take the message forwarding decision. Depending on the local information needed, we call a local routing \emph{source oblivious} if it does not need to know the  coordinates of the message source, or \emph{$0$-bit} if the algorithm does not encode any information in the message header except for the target address, or {$t$-local} if the local information contains the coordinates of the neighbors that are within $t$ hop distance from the current vertex. The \emph{competitive ratio} of a path $P$, taken by the algorithm, is the sum of the edge lengths of $P$ over the length of the shortest source-to-destination path.
\marginnote{Do we really need any of this?}
}

\section{Angle-Monotone Graphs of Width $90^\circ$}
\label{sec:subq}

In this section we show that any set of $n$ points admits an angle-monotone graph of width $90^\circ$ with $o(n^2)$ edges.

To achieve this, we will use the   Erd\H{o}s-Szekeres  theorem~\cite{Erdos-Szekeres} to 
partition the point set into
 subsets each with a logarithmic number of points in convex position.  We will then construct an angle-monotone graph on each pair of subsets.
 Our construction is inspired by and builds upon a result in~\cite{D-Frati-G} that every `one-sided convex point set' admits an 
increasing-chord graph with a linear number of edges.  In fact, their proof yields an angle-monotone graph of width $90^\circ$ as we explain in the following section. 

\subsection{Angle-Monotone Graphs on Convex Point Sets}
\label{sec:deh}
Dehkordi et al.~\cite{D-Frati-G}  showed that every convex point set of $n$ points admits an  angle-monotone graph with $O(n \log n)$ edges.  
\changed{Actually, 
they only state that there is an increasing-chord graph.} 
Here we explain why their proof gives the stronger result we need. 
 
A point set $P$ is called \emph{one-sided} with respect to some directed straight line $\vec{d}$, which is not orthogonal to any line through two points of $P$, if the order of the projections of the points on $\vec{d}$ corresponds to the order the points   on the convex-hull of $P$. Figure~\ref{fig:deh}(a) illustrates a one-sided point set. Given a one-sided point set $P$ with respect to the positive $x$-axis,  Dehkordi et al.~\cite{D-Frati-G} showed how to construct a spanning increasing-chord graph $G$ on $P$ with $O(n)$ edges such that any pair of vertices in $G$ are connected by an $xy$-monotone path. Since $xy$-monotone paths are angle monotone~\cite{AlamdariCGLP12}, $G$ is an angle-monotone graph. 

Dehkordi et al.~\cite{D-Frati-G} used the technique for one-sided point sets recursively to construct increasing-chord graphs on arbitrary convex point sets. They showed that any convex point set $P$ can be partitioned into four one-sided point sets $P_1,P_2,P_3,P_4$, e.g., see Figure~\ref{fig:deh}(b), such that the following properties hold: (a) The points in each $P_i$, where $1\le i\le 4$, appear consecutively on the convex hull of $P$. (b) The partition is balanced, i.e., $\max\{|P_1|+|P_3|, |P_2|+|P_4|\} \le n/2+1$. (c)  The point sets $(P_1\cup P_2)$, $(P_2\cup P_3)$, $(P_3\cup P_4)$, and $(P_4\cup P_1)$ are one-sided.

\begin{figure}[h]
\centering
\includegraphics[width=.85\textwidth]{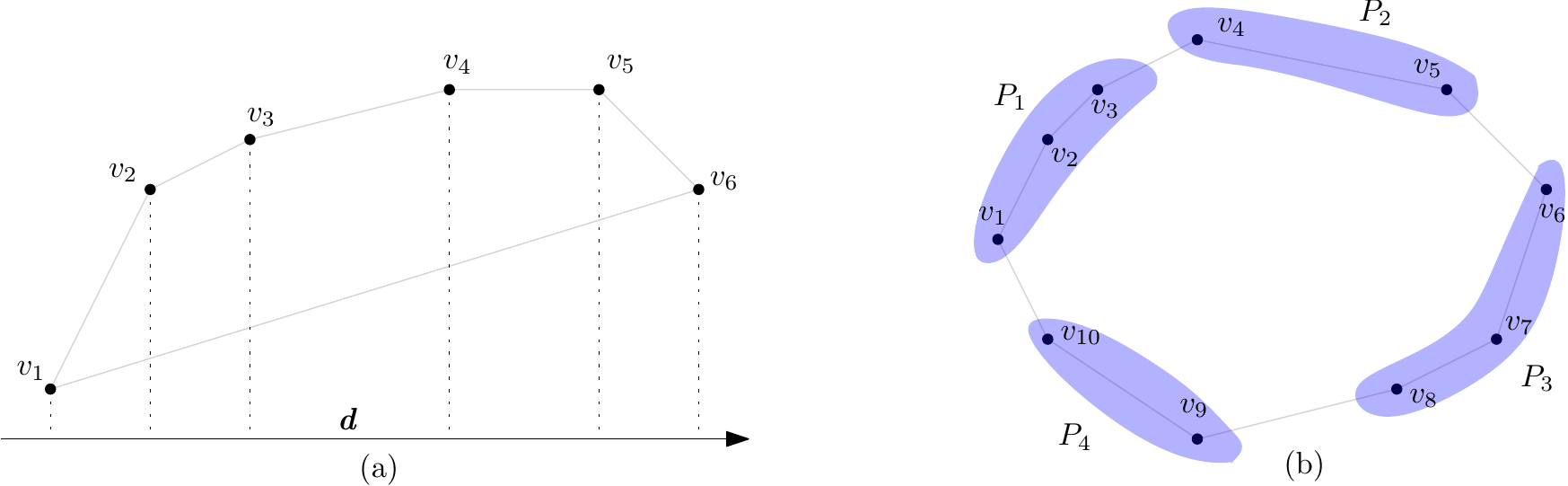}
\caption{(a) A one-sided convex point set. (b) Illustration for the construction of angle-monotone graphs.}
\label{fig:deh}
\end{figure}

Consequently, one can first  construct linear size angle-monotone graphs for one-sided point sets $(P_1\cup P_2)$, $(P_2\cup P_3)$, $(P_3\cup P_4)$, $(P_4\cup P_1)$, and then recursively construct angle-monotone graphs for the convex point sets $(P_1\cup P_3)$ and  $(P_4\cup P_1)$. The union of all these graphs contains angle monotone paths for every pair of vertices. Since $\max\{|P_1|+|P_3|, |P_2|+|P_4|\} \le n/2+1$, the size of the final angle-monotone graph is $f(n) \le 2\cdot f(\frac{n}{2}+1) + O(n) \in O(n \log n)$.

\subsection{Angle-Monotone Graphs on Arbitrary Point Sets}


We first introduce some  preliminary definitions and notation. 
 We will distinguish two types of $x$-monotone paths: an \emph{$(x,y)$-monotone path} increases in both $x$- and $y$-coordinates, and an \emph{$(x,-y)$-monotone path} increases in $x$-coordinate and decreases in $y$-coordinate. For each type of path we further distinguish convex and concave subtypes.  Traversed in increasing $x$ order, a convex path turns to the right, and a concave path turns to the left. Thus \emph{an $(x,y)$-convex path} is an $(x,y)$-monotone path that turns to the right when traversed in increasing $x$ order, and etc.~for the other three types. See Figures~\ref{fig:surprizing}(a)--(d).

  
\begin{figure}[htb]
\centering
\includegraphics[width=.8\textwidth]{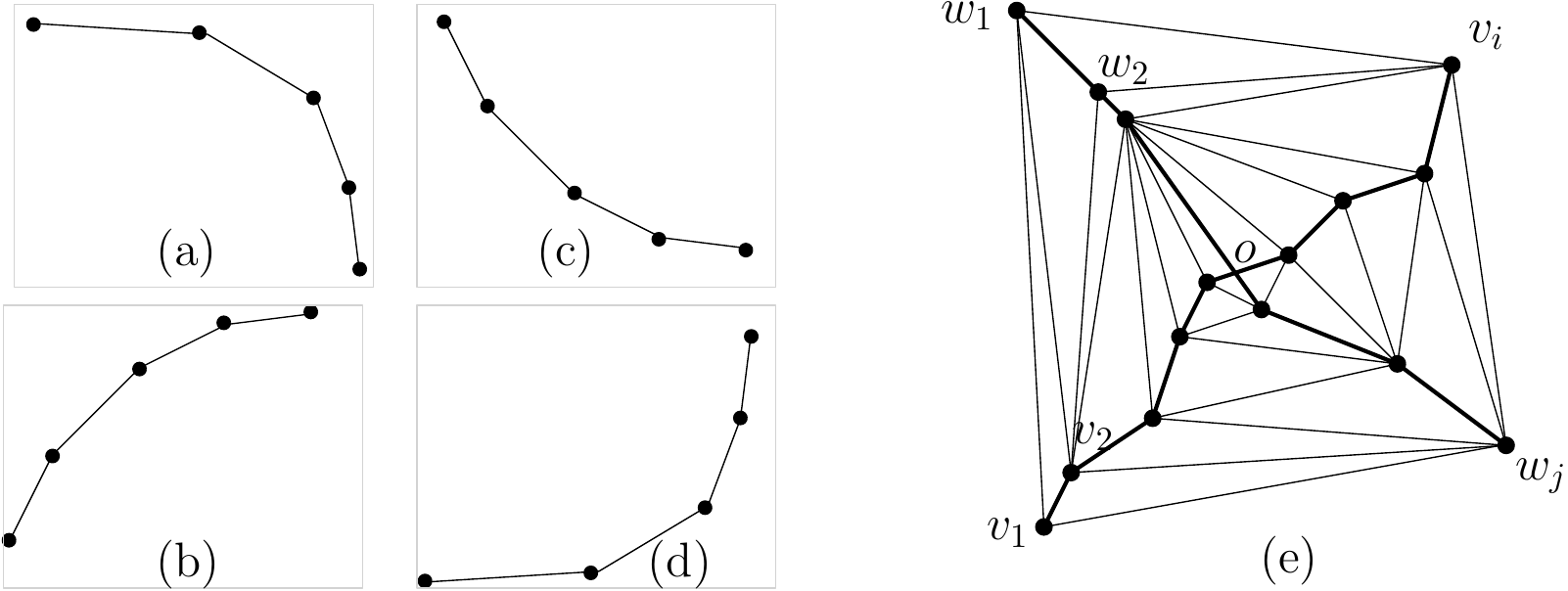}
\caption{(a) An $(x,-y)$-convex path. 
(b) An $(x,y)$-convex path. 
(c) An $(x,-y)$-concave path. 
(d) An $(x,y)$-concave path. (e) Illustration for Lemma~\ref{lem:surprizing}.}
\label{fig:surprizing}
\end{figure}

\begin{lemma}
\label{lem:surprizing} 
 Let $P = (v_1,\ldots,v_i)$ be an $(x,-y)$-monotone path, and let $P' = (w_1,\ldots,w_j)$ be an $(x,y)$-monotone path. Then there exists an angle-monotone graph  of width $90^\circ$ and size $O(i+j)$ that spans 
 $P$ and $P'$. 
\end{lemma} 
\begin{proof}
Assume without loss of generality that $P$ and $P'$ intersect, say at point $o$.  (If necessary, we can add points $(-\infty, \infty)$ and $(\infty, -\infty)$ at the start and end of $P$ respectively, and similarly for $P'$.)
We will solve four subproblems for the points to the left of $o$, to the right of $o$, above $o$ and below $o$, as 
illustrated in Figure~\ref{fig:surprizing}(e).  Observe that any two points in $P \cup P'$ either lie in the same path, or in one of these half-spaces, so it suffices to find an angle-monotone graph of size $O(i+j)$ for each subproblem, and take the union. 

 
 Let $v_1,\ldots,v_{i'}$ and $w_1,\ldots,w_{j'}$ be the vertices
 to the left of the vertical line through $o$.
 We now construct an angle-monotone graph 
 spanning these vertices as follows.  Add an edge  $(v_1,w_1)$ and then move a vertical sweep-line $\ell$ from $(-\infty,0)$ to $o$. 
 Each time we encounter a new vertex $q$, we add the edges 
 $(q,v')$ and $(q,w')$, where $v'$ (resp., $w'$) is the rightmost vertex of $P$ (resp.,  $P'$) lying in the left-half plane of $\ell$. We call $v'$ and $w'$ the \emph{predecessor} of $q$ in $P$ and in $P'$, respectively.
 The resulting graph $H$ has size $O(i+j)$.  
 We now show  that $H$ is an angle-monotone graph.
 For any pair of vertices $a,b$,
 if  $a,b$ belong to the same path, i.e., $P$ or $P'$, then they
 are already connected by an angle-monotone path. 
 Otherwise, assume without loss of generality that 
 $a\in P$, $b\in P'$, and 
 $b$ has a 
 larger $x$-coordinate than $a$.
 Let $b'$ be the predecessor of $b$ in $P$.  Follow the path $P$ from $a$ to $b'$ and then take the edge $(b',b)$.  This is an 
$(x,-y)$-monotone path, and thus angle-monotone (equivalently, of width $90^\circ$). 
%
\end{proof}

\begin{lemma}
\label{lem:convex}
 Let $P = (v_1,\ldots,v_i)$ be an $(x,-y)$-convex  path, and let $R$ be the region (above $P$) bounded by $P$ and the leftward and downward rays starting at $v_1$ and $v_i$, respectively. Then for any set 
 $W$ of $j$ points in $R$,
  there exists a graph $G$ 
 of size $O(i+j)$ such that any pair of vertices $v\in P, w\in W $ is connected by an angle-monotone path of width $90^\circ$.
\end{lemma}
\begin{proof}
 Let $v_0$ be any point on the leftward ray starting at $v_1$. For each $q$ from $1$ to $i$, let $\ell_q$ be the 
  ray starting at $v_q$ that lies perpendicular to $v_{q-1}v_q$ and enters region $R$.
  Since $P$ is convex, the rays  $\ell_q$  subdivide the region $R$ into  regions $R_0,R_1,\ldots,R_i$, e.g., see Figure~\ref{fig:convex}(a). For each point $v_q$, connect $v_q$  to all the points in region $(R_{q}\cap W)$, e.g., see Figure~\ref{fig:convex}(b). Let $G'$ be the resulting graph including the edges of $P$. 
 We now claim that for any vertex $v_t$, $1 \le t \le q$ and for any $w\in (R_{q}\cap W)$ the path $v_t, \ldots, v_q, w$ is an angle-monotone path. 
  If the $y$-coordinate of $w$ is smaller than that of $v_q$, then 
 this path is $(x,-y)$-monotone and hence angle-monotone,
 e.g., see Figure~\ref{fig:convex}(b).
 Otherwise,  one can observe that all edges in the path have vectors that lie in the $90^\circ$ clockwise wedge between $\ell_q$ and the line extending $(v_{q-1},v_q)$, e.g., see Figure~\ref{fig:convex}(c).   Thus the path $v_t, \ldots, v_q, w$ is an angle-monotone path. 
%

\begin{figure}[pt]
\centering
\includegraphics[width=\textwidth]{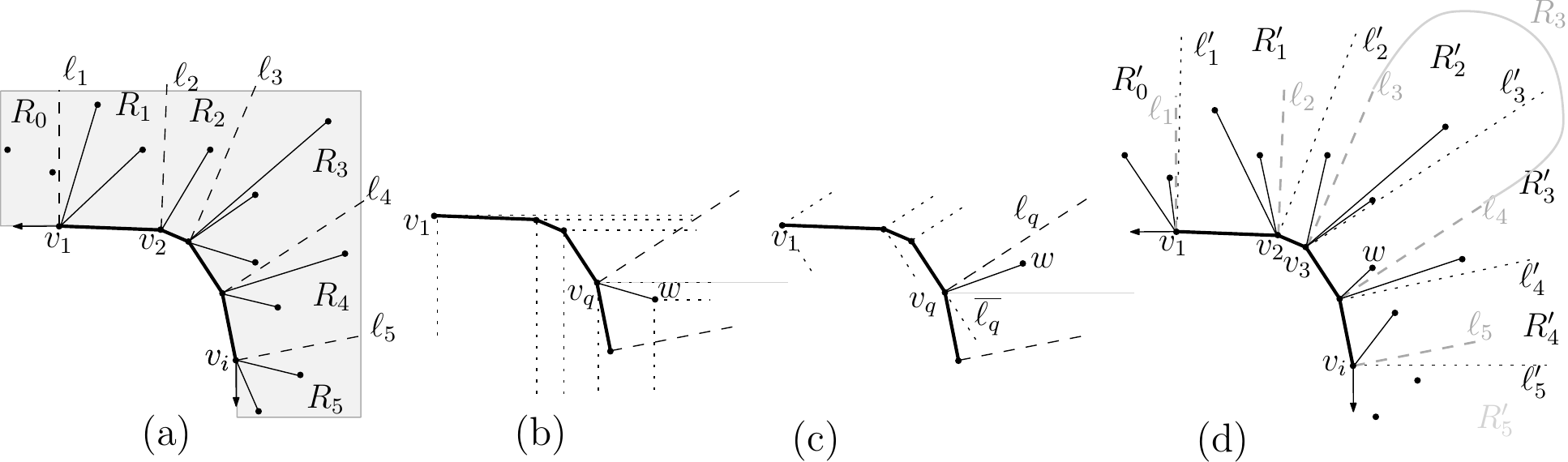}
\caption{(a)--(d) Illustration for Lemma~\ref{lem:convex}. 
}
\label{fig:convex}
\end{figure}

For each $q$ from $i$ to $1$, we construct a graph $G''$ symmetrically by  defining the perpendicular rays  $\ell'_1,\ldots,\ell'_i$ and  regions $R'_0,\ldots,R'_{i}$, as illustrated in Figure~\ref{fig:convex}(d).   We construct the  final graph $G$ by taking the union of all the edges of $G'$ and $G''$. It is straightforward to observe that $G$ has at most $(i+2j)$ edges. 

To complete the proof, we must show that there is an angle-monotone path from any vertex $v_t$, $1 \le t \le i$ to any $w \in W$. Observe that $R_q$ and $R'_{q-1}$ intersect because $P$ is convex. 
If $w \in (R_t \cup \cdots \cup R_i)$, then there is an angle-monotone path from $v_t$ to $w$ in $G$, and otherwise $w \in (R'_{t-1} \cup \cdots  \cup R'_0)$ and there is an angle-monotone path from $v_t$ to $w$ in $G''$.
%
%
\end{proof}

\begin{lemma}
\label{lem:br}
 Let $P = (v_1,\ldots,v_i)$ and $P' = (w_1,\ldots,w_j)$ be a
 pair of $(x,-y)$-convex (or, concave) paths. Then there exists an angle-monotone graph  (spanning  $P$ and $P'$) with width $90^\circ$ and size $O(i+j)$.
\end{lemma} 
\begin{proof} We prove the lemma assuming that $P$ and $P'$  are a pair of convex paths.  The case when they are concave is symmetric. We consider two cases depending on whether $P$ and $P'$ intersect or not.
 
\textbf{Case 1:} First consider the case when  $P$ and $P'$ do not intersect, and 
 assume without loss of generality that $P'$ lies above $P$. Since the vertices on $P'$ are already connected by an angle-monotone path, 
 we can apply Lemma~\ref{lem:convex} to obtain the required angle-monotone graph.

\textbf{Case 2:} Consider now the case when $P$ and $P'$ intersect. Let $o_1,\ldots, o_t$ be the points  of intersections   ordered from left to right, e.g., see Figure~\ref{fig:convex2}(a).
 Let $A_1$ (resp., $A_{t+1}$) be the set of  vertices of  $(P\cup P')$ with $x$-coordinates smaller (resp., larger) than that of $o_1$ (resp., $o_t$). For every $q$, where $2\le q\le t$, let $A_q$ be the set of vertices of   $(P\cup P')$  that lie to the left of $o_q$ and to the right of $o_{q-1}$. 

We process the sets $A_1,\ldots, A_{t+1}$ independently using Case 1, and let $G_{A_1},$ $ \ldots, G_{A_{t+1}}$ be the resulting graphs. Compute the final graph $G$ by taking the union of $P,P'$, and $G_{A_1}, \ldots, G_{A_{t+1}}$. It is straightforward to verify that every pair of vertices in $G$ is connected by an angle-monotone path. 
The number of edges in $G$ is at most $\sum_{k=1}^q |A_k|   \in O(i+j)$.
\end{proof}

\begin{figure}[pt]
\centering
\includegraphics[width=.7\textwidth]{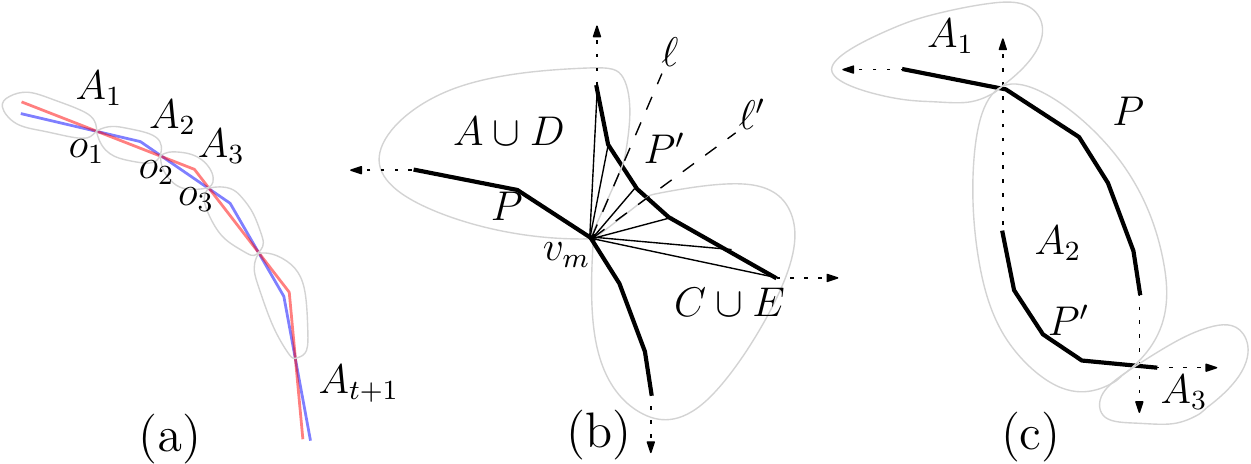}
\caption{  (a)--(c) Illustration for Lemmas~\ref{lem:br}--\ref{lem:br2}. }
\label{fig:convex2}
\end{figure}


\begin{lemma}
\label{lem:br2}
 Let $P = (v_1,\ldots,v_i)$ be an $(x,-y)$-convex path, and let 
   $P' = (w_1,\ldots,w_j)$ be an $(x,-y)$-concave path (or, vice versa).  Then there exists an angle-monotone graph  (spanning  $P$ and $P'$) of width $90^\circ$ and size $O(k\log k )$, where $k = \max\{i,j\}$.  
\end{lemma} 
\begin{proof}
 We extend $P$ by adding leftward and downward rays starting at $v_1$ and $v_i$, respectively, e.g., see Figure~\ref{fig:convex2}(b). 
 We extend $P'$ symmetrically. We now consider two cases depending on whether $P,P'$ intersect or not. 

\textbf{Case A:} If $P'$ and $P$ do not intersect,  e.g., see Figure~\ref{fig:convex}(f), then $P'$ lies above $P$. In this scenario we can find  an angle-monotone graph  of 
size $O(k)$ 
 by applying Lemma~\ref{lem:convex}.

\textbf{Case B:} If $P$ and $P'$  intersect, then they 
 intersect in at most two points $o_1, o_2$, with $o_1$ to the left of $o_2$, e.g., see Figure~\ref{fig:convex2}(c). 
 The part to the left of $o_1$ and the part to the right of $o_2$ can be handled using Case A.  In the middle we have a convex polygon, where the result of Dehkordi et al.~\cite{D-Frati-G} gives  an \AMG of size $O(k \log k)$ (see Section~\ref{sec:deh}). 
\end{proof}

\begin{theorem}
Let $S$ be a point set with $n$ points. Then there exists an angle-monotone graph   (spanning $S$) of width $90^\circ$ and size $O(\frac{n^2 \log \log n}{\log n})$ edges. 
\end{theorem}
\begin{proof}
By the Erd\H{o}s-Szekeres     theorem~\cite{Erdos-Szekeres}, every point set  with $n$ points contains  a
 subset of $O(\log n)$ points in convex position.  Urabe~\cite{Urabe96} observed that by repeatedly extracting such 
 a convex set, one can partition a point set into $O(\frac{n}{\log n})$ convex  polygons each of size $O(\log n)$. 
 We partition each of these convex polygons into  
 an $(x,y)$-convex path, an $(x,-y)$-convex path, an $(x,y)$-concave path,
  and an $(-x,-y)$-concave path.

For each pair of these 
 paths, we apply Lemmas~\ref{lem:surprizing}--\ref{lem:br2}, as appropriate. Finally, we compute the required graph $G$ by taking the union of all the $O(\frac{n^2}{\log ^2 n})$ graphs. Since  any pair of points in $S$  either  lie on the same  path,  or in one of these $O(\frac{n^2}{\log ^2 n})$ graphs, they are connected by an angle-monotone path of width $90^\circ$. Since the length of each path is at most $O(\log n)$, the size of $G$ is $O(\frac{n^2}{\log ^2 n}) \cdot O(\log n \log \log n) = $  $O(\frac{n^2 \log \log n}{\log n})$. 
\end{proof}


\begin{cor}
Let $S$ be a point set with $t$ nested convex hulls. Then there exists an angle-monotone graph (spanning $S$) of width $90^\circ$ with $O(t^2  n \log n)$ edges.
\end{cor}

\subsection{Further Observations}
Although the \changed{above} construction of a subquadratic-size angle-monotone network with width $90^\circ$ is \changed{somewhat} involved, 
one can easily construct an \AMG with width $(90^\circ+\alpha)$  and $O(\frac{n^{3/2}}{\alpha})$ edges, for any $0<\alpha \le 90^\circ$, as 
\changed{we show in this section.  In 
the following Section~\ref{sec:small-size} we give a different construction to obtain a graph with $O(\frac{n}{\alpha})$ edges.}

Let $S$ be a set of $n$ points in $\mathbb{R}^2$.  
  To construct an angle-monotone graph, we first 
  mark a set $R$ of $\sqrt{n}$ points from $S$, and for each pair 
  of points $(a,b)$, where $a\in S$ and $b\in S$, we construct 
  an angle-monotone path of width $90^\circ+\alpha$ between $a$ and $b$.
  We then apply this process recursively on $S\setminus R$. 
  We now describe the construction in details. Assume initially
  all the points of $S$ are unmarked.

  Mark a set $R$ of $\sqrt{n}$ points from the unmarked points of $S$.
  Let the set of unmarked points be $R'$. 
  Construct a clique $K_{|R|}$ spanning the points of $S$.
  For each point $q\in R'$, create $t = 360^\circ/(2\alpha)$ uniform
  wedges of angle $2\alpha$ around $q$. Figure~\ref{fig:n3by2}(a)
  illustrates an example, where the points of $R$ are shown in black.
  For each wedge $W$,   let $W(R)$ be the points of $R$ that lie inside $W$. Add an edge between $q$ and the bisector nearest neighbor 
  of $q$ in $W(R)$. Let the resulting graph be $H$. 
  Note that $H$ has $O(nt)\in O(n/\alpha)$ edges.

\begin{figure}[h]
\centering
\includegraphics[width=.85\textwidth]{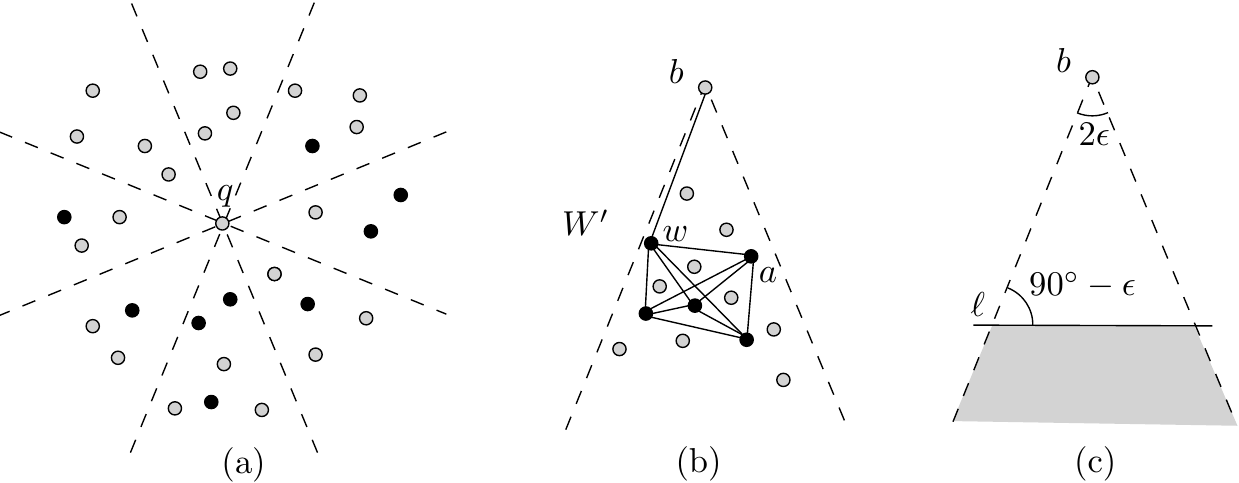}
\caption{Construction of angle-monotone graphs.}
\label{fig:n3by2}
\end{figure}

  We claim that for each pair 
  of points $(a,b)$, where $a\in R$ and $b\in S$, $H$ contains 
  an angle-monotone path of width $90^\circ+\alpha$ between $a$ and $b$. 
  If $a\in R$ and $b\in R$, then the claim is straightforward to verify.
  If $a\in R$ and $b\in R'$, then let $W'$ be the wedge of 
  $b$ that contains $a$, e.g. see Figure~\ref{fig:n3by2}(b).
  Since the points in $R$ form a clique in $H$,
  the points of $W'(R)$ form  a clique inside $W'$. Let $w\in W'(R)$ be 
  bisector nearest neighbor of $b$ in $W'$.
  If $w$ coincides with $a$, then $(a,b)$ must be an edge in $H$.
  We may thus assume that $w\not=a$. In this scenario, 
  the smallest angle determined by  the path
  $a,w,b$ is at least $(90^\circ-\alpha)$. 
  Therefore, $a,w,b$ is an angle-monotone path of width 
  at most $180^\circ-(90^\circ-\alpha) =  (90^\circ+\alpha)$.
  Figure~\ref{fig:n3by2}(c) illustrates such a scenario, where   
  the line $\ell$ passes through $w$ and perpendicular to the 
  bisector of $W'$. The region where $a$ could be located is shown in 
  gray. 
  
 We now apply the above process repeatedly until we mark all the points of $S$.  Since at each step we process $\sqrt{n}$ new points, 
 the number of steps is $O(\sqrt{n})$. Since at each step we 
 create at most $O(n/\alpha)$ edges, the number of total 
 edges is bounded by $O(\frac{n^{3/2}}{\alpha})$. The resulting \AMG  has diameter 2.

In the following section we give a more interesting construction of an \AMG of width $(90^\circ+\alpha)$. 


\section{Angle-Monotone Graphs of Width $(90^\circ+\alpha)$} 
\label{sec:small-size}

In this section we show how to construct, for any point set of size $n$ and any angle $\alpha$, $0 < \alpha < 45^\circ$, an angle-monotone graph of width $(90^\circ+\alpha)$ with $O(\frac{n}{\alpha})$ edges. 
We call these
\emph{layered 3-sweep graphs}.
First, in Section~\ref{sec:basic-graph}, we introduce a \emph{3-sweep graph} of a point set in which three lines are used to connect each point to three of its neighbors. The special case where the three lines form $60^\circ$ wedges yields the half-$\Theta_6$-graph.  In Section~\ref{sec:properties},  we analyze angle-monotonicity properties of 3-sweep graphs. Then, in Section~\ref{sec:full-graph}, we define a \emph{$k$-layer 3-sweep graph} as a union of $k$ different 3-sweep graphs. We prove that a layered 3-sweep graph with an appropriate number of layers is an angle-monotone graph of width $(90^\circ+\alpha)$ with  $O(\frac{n}{\alpha})$ edges. 


\subsection{3-Sweep Graphs}
\label{sec:basic-graph}

\begin{figure}[pt]
\centering
\includegraphics[width=.9\textwidth]{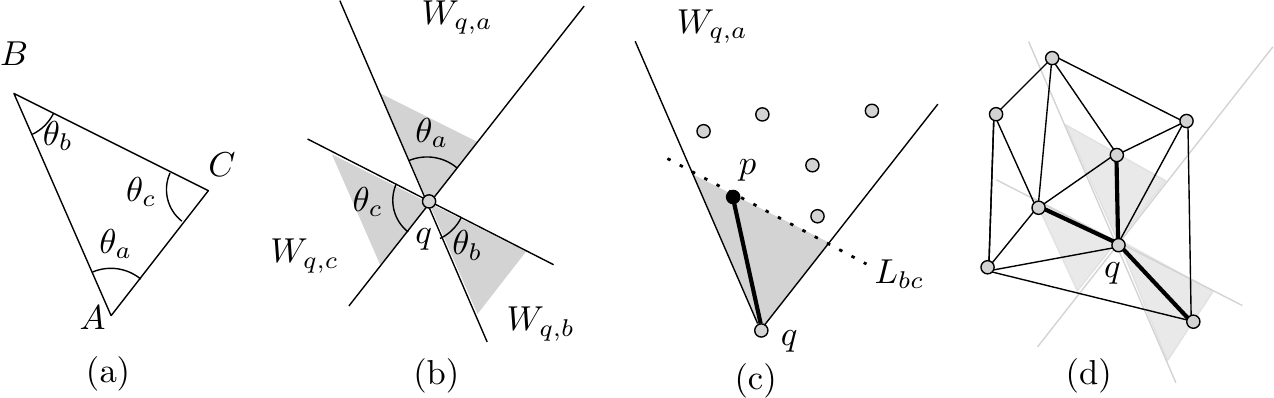}
\caption{(a)  $\Delta ABC$. (b) $W_{q,a}, W_{q,b}, W_{q,c}$. 
 (c)  The $a$-nearest neighbor of $q$, where $(q,p)$ is a $\theta_a$-edge. (d) A 3-sweep graph.}
\label{fig:nnshort}
\end{figure}
 
 Let $\Delta ABC$ be an acute triangle in  $\mathbb{R}^2$ such that 
 $A,B,C$ appear in clockwise order on the perimeter of  $\Delta ABC$, e.g., see Figure~\ref{fig:nnshort}(a).  Let $\theta_a$, $\theta_b$, $\theta_c$ be the angles  at $A,B,C$, respectively.
 For any point $q$ let $W_{q,a}$ (the ``$a$-wedge'' of $q$) be the wedge with apex $q$ such that the two sides  of $W_{q,a}$ are parallel to $AB$ and $AC$, i.e., $\Delta ABC$ can be  translated such that $A$ coincides with $q$ and 
 two sides of $\Delta ABC$ lie along the sides of $W_{q,a}$.
 Similarly, we define the wedges $W_{q,b}$ and $W_{q,c}$, e.g., see Figure~\ref{fig:nnshort}(b). 
 The \emph{$a$-nearest neighbor} of $q$ in $W_{q,a}$ is defined to be the first point $p$ that we hit (after $q$) while sweeping $W_{q,a}$ by a line $L_{bc}$ parallel to $BC$ (starting with the line through $q$).  Figure~\ref{fig:nnshort}(c) illustrates such an example. In the case of ties, we can pick arbitrarily as far as the results in this subsection are concerned.  However, it is important that 
the local routing algorithm in Section~\ref{sec:local-routing} be able to find the $a$-nearest neighbor, so we break ties by choosing the most clockwise point. 
 We call the edge $(q,p)$ a \emph{$\theta_a$-edge}. We define 
 $b$- and $c$-nearest neighbors and $\theta_b$- and $\theta_c$-edges analogously.
 %

Given a set of points $S$, and three acute angles $\theta_a, \theta_b, \theta_c$ summing to $180^\circ$, we define a \emph{3-sweep graph  
 $G$ on $S$ with angles $\{\theta_a,\theta_b,\theta_c\}$}
 to be a geometric graph obtained by connecting every point $q\in S$ to its $a$-, $b$- and $c$-nearest neighbors, e.g., see  Figure~\ref{fig:nnshort}(d).
 If $\theta_a = \theta_b = \theta_c = 60^\circ$, then 
 $G$ is equivalent to the well known half-$\Theta_6$-graph.
 
 Bonichon et al.~\cite{bonichonTD} proved that \changed{half-$\Theta_6$-graphs are equivalent to 
 Triangular Distance (TD) Delaunay triangulations, introduced by Chew~\cite{Chew86}.} 
 A 3-sweep graph is also the same as the  half-$\Theta_6$-graph under a linear transformation. 
 As illustrated in Figure~\ref{fig:tddel}, the linear transformation that maps an equilateral triangle $T$ into $T' (=\Delta ABC)$ transforms point set $S$ into $S'$ so that the half-$\Theta_6$  graph on $S$ maps to the 3-sweep graph on $S'$.


\begin{figure}[h]
\centering
\includegraphics[width=.6\textwidth]{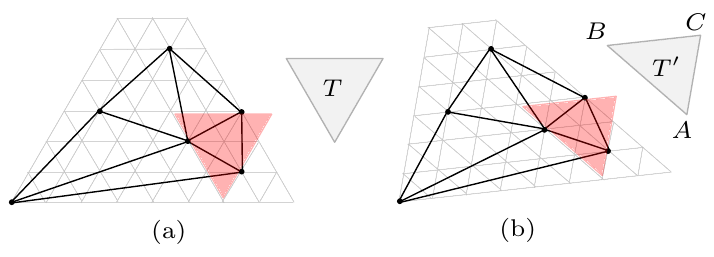}
\caption{(a) A TD Delaunay triangulation (equivalently, a half-$\Theta_6$-graph) of a point set $S$, where the distance function is determined by  the equilateral triangle $T$. (b) A 3-sweep graph on a point set $S'$ determined by a triangle $T'$, where both $S'$ and $T'$ are transformed using the same linear transformation. }
\label{fig:tddel}
\end{figure}

Both half-$\Theta_6$ and 3-sweep graphs  are special cases of \emph{convex Delaunay graphs}, which were studied by Bose et al.~\cite{BoseCCS10}. They proved that every convex Delaunay graph is a $t$-spanner,  but the value of $t$ obtained from that proof is too large to be useful for our triangle $T'$. 
\changed{In particular,} the constant $t$ depends on two parameters   
$\alpha_c, \kappa_c$ of the underlying convex shape.
For half-$\Theta_6$-graphs, the convex shape is an equilateral triangle, e.g., see Figure~\ref{fig:tddel2}, and $t$ is bounded by 58, which is much larger than the known spanning ratio of 2 for half-$\Theta_6$-graphs~\cite{bonichonTD,Chew86}.

\begin{figure}[h]
\centering
\includegraphics[width=.3\textwidth]{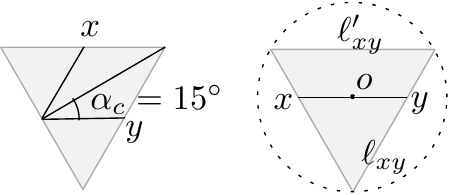} 
\caption{Illustration for $\alpha_c = 15^\circ$ and $\kappa_c = (\ell'_{xy}/|xy|) = (5/2)$ for an equilateral triangle.
}
\label{fig:tddel2}
\end{figure}

Every convex Delaunay graph is planar~\cite{BoseCCS10}, and hence the following lemma is immediate. For interest, here we give a  self-contained proof. 


\begin{lemma}
\label{lem:planar}
Every 3-sweep graph is planar.
\end{lemma}
\begin{proof}
Note that it suffices to prove the following claim.
\begin{enumerate}
\item[] Let $S$ be a set of points in $\mathbb{R}^2$, and let $q$ and $t$ be two points in $S$. Let $q'$ be a nearest neighbor of $q$ in  $W_{q,a},W_{q,b},$ or $W_{q,c}$. Similarly, let $t'$ be a nearest neighbor of $t$ in $W_{t,a},W_{t,b}$, or $W_{t,c}$. Then the line segments $qq'$ and $tt'$ do not intersect except possibly at their common endpoint, i.e., when $q'=t'$. 
\end{enumerate}
The case when $q'\in W_{q,j}$ and $t'\in W_{t,j}$, for some $j\in \{a,b,c\}$ is straightforward.  We may thus assume without loss of generality that $q'\in W_{q,a}$ and $t'\in W_{t,b}$, e.g., see Figure~\ref{fig:nn}(a).  We now show that the line segments $qq'$ and $tt'$ do not intersect except possibly at their common endpoint, i.e., when $q'=t'$. 

Suppose for a contradiction that  there exist 
 $q,q',t,t'$ such that the segments $qq'$ and $tt'$ properly intersect.
 Let $r$ be the point of intersection. Since both 
 $W_{q,\theta_a}$ and $W_{t,\theta_b}$ contains $r$, 
 either $q\in W_{t,\theta_b}$ or $t\in W_{q,\theta_a}$.
 
 Without loss of generality assume that $t\in W_{q,\theta_a}$.
 Let  $L$  be the straight line that passes through $q'$
 and  makes a clockwise angle of $\theta_b$ with the left side of $W_{q,\theta_a}$,  e.g., see Figure~\ref{fig:nn}(b). Since $q'$ is a nearest neighbor of $q$, the point $t$ must be on or above  $L$ . We now consider two cases depending
 on whether  $t'$ is inside or outside of $W_{q,\theta_a}$.

\begin{figure}[pt]
\centering
\includegraphics[width=.9\textwidth]{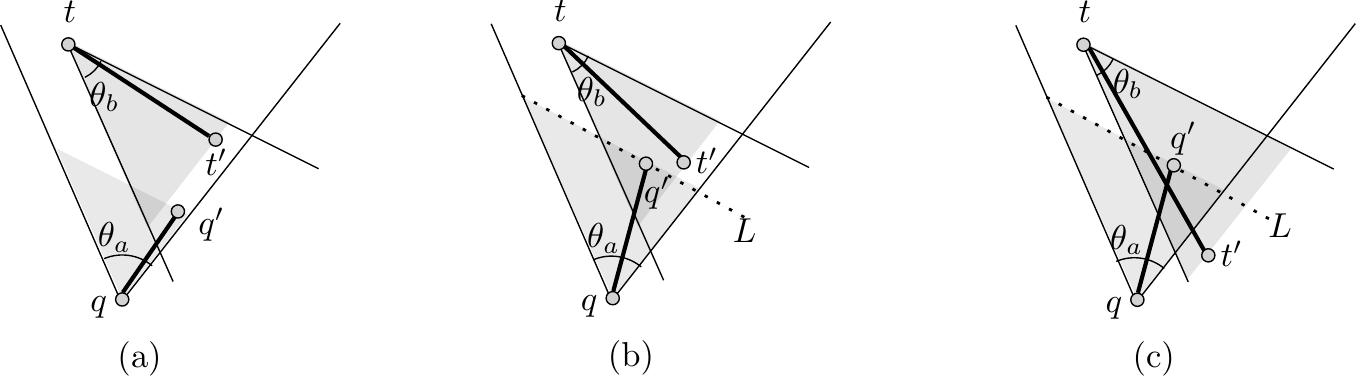}
\caption{Illustration for the proof of Lemma~\ref{lem:planar}.}
\label{fig:nn}
\end{figure}
 
 If $t'\in W_{q,\theta_a}$, then $t'$ must be on or above  $L$ . Consequently, $tt'$ may intersect $qq'$ only if $t,q'$ and $t'$ lie on  $L$  in this order. Since $t'$ is a nearest neighbor of $t$, the point $q'$ cannot have smaller  distance 
 to $t$ than that of $t'$. Hence $q'$ must coincide with $t'$, a contradiction.
   
 If $t'\not\in W_{q,\theta_a}$, then $t'$ must lie to the right of the right side of $W_{q,\theta_a}$, e.g., see Figure~\ref{fig:nn}(c). Since $t$ lies on or above  $L$  and since $qq'$ intersects $tt'$, $q'$ must lie inside $W_{t,\theta_b}$. Consequently, $q'$ must have smaller  
 distance to $t$ than that of $t'$, a contradiction.
\end{proof}

\subsection{Angle-Monotonicity of 3-Sweep Graphs}
\label{sec:properties}
We now analyze angle-monotonicity properties of 3-sweep graphs.
We will show that for points $q$ and $t$ in a 3-sweep graph $G$ with $t$ in $W_{q,a}$ there is an angle-monotone path from $q$ to $t$ whose width depends on $\theta_a$ and on the position of $t$ relative to the \emph{$a$-path} of $q$.  The \emph{$a$-path} of $q$, denoted $P_{q,a}$, is defined to be the maximal path $q(=v_0),\ldots, v_k$ in $G$ such that for each $i$ from $1$ to $k$, $v_i$ is the  $a$-nearest neighbor of $v_{i-1}$. 
We also define the \emph{extended $a$-path} $\overline P_{q,a}$ to be the $a$-path $P_{q,a}$ together with $W_{v_k,a}$, which is empty of points since the $a$-path is maximal.  
We define [extended] $b$- and $c$-paths similarly.

Observe that if $t$ is a vertex of $P_{q,a}$ then there is an angle-monotone path of width $\theta_a$ from $q$ to $t$.  The following lemma handles the case where $t$ $\in W_{q,a}$, and $t$ does not lie on the $a$-path from $q$.  
The proof of the lemma is very similar to the proof in~\cite{angle-mono} that the half-$\Theta_6$-graph is angle-monotone of width $120^\circ$.


\begin{lemma}
\label{lem:angle-bound}
Let $q$ and $t$ be two vertices in $G$ such that $t$ lies in $W_{q,a}$.  If $t$ lies to the left (resp., right) of $\overline P_{q,a}$ then there is an angle-monotone path of width $(\theta_a+\theta_b)$ (resp., $(\theta_a+\theta_c)$) from $q$ to $t$. 
 Furthermore, the path consists of one subpath of the $a$-path of $q$ followed by one subpath of the $b$-path (resp., $c$-path) of $t$. 
\end{lemma}
\begin{proof}
%
 Assume that the side $BC$ of $\Delta ABC$ is parallel to the $x$-axis and $A$ lies below $BC$. Such a condition can be met after a suitable rotation of the Cartesian axes. Without loss of generality assume that  $t$ lies to the left of $\overline P_{q,a}$.  We will show that $P_{q,a}$ and $P_{t,b}$ intersect at some vertex $x$. Our path will then follow $P_{q,a}$ from $q$ to $x$, and then follow $P_{t,b}$ backwards from $x$ to $t$. Observe that this path is angle-monotone of width $(\theta_a+\theta_b)$.

Our proof is by contradiction.  Assume that $P_{q,a}$ and $P_{t,b}$ do not intersect at a vertex. Let $t'$ be the last vertex of $P_{t,b}$ that lies in $W_{q,a}$ and strictly to the left of $\overline P_{q,a}$.
Let $q'$ be the last vertex of $P_{q,a}$ that lies below or at the same height (i.e., $y$-coordinate) as $t'$.

We will derive a contradiction by considering the possible positions for $t'$ and $q'$.
First suppose that $t'$ is in $W_{q',a}$.  See Figure~\ref{fig:tddel-new2}(a).
Then $q'$ must have an $a$-nearest neighbor $q''$, since $t'$ is a candidate to be its $a$-nearest neighbor.  By definition of the $a$-nearest neighbor, $q''$ must be at the same height as $t'$, or lower. Note that $q''$ is on $P_{q,a}$.  This contradicts the choice of $q'$ as the last vertex of $P_{q,a}$ that lies below or at the same height as $t'$.

\begin{figure}[h]
\centering
\includegraphics[width=.4\linewidth]{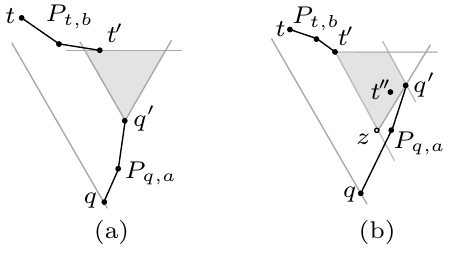}
\caption{ (a) The case of the proof of Lemma~\ref{lem:angle-bound} when $t' \in W_{q',a}$.  (b) The case of the proof of  Lemma~\ref{lem:angle-bound} when $t' \not\in W_{q',a}$. 
}
\label{fig:tddel-new2}
\end{figure}

Next suppose that $t'$ is not in $W_{q',a}$. See Figure~\ref{fig:tddel-new2}(b). Consequently, $q'$ must be in $W_{t',b}$. Then $t'$ must have a $b$-nearest neighbor $t''$, since $q'$ is a candidate to be its $b$-nearest neighbor. By definition of the $b$-nearest neighbor, $t''$ must be to the left of, or on, the line $\ell$ parallel to AC going through $q'$.  Thus $t''$ is in $W_{z,a}$ where $z$ is the point where line $\ell$ intersects the line forming the left side of $W_{t',b}$.
If $t''$ is in $W_{q',a}$ then (as argued above) $q'$ must have an $a$-nearest neighbor that provides a better choice than $q'$. 

Thus $t''$ must lie in the quadrilateral $W_{z,a} - W_{q',a}$ (shaded in Figure~\ref{fig:tddel-new2}(b)).  
Observe that $t''$ lies in $W_{q,a}$ since both $t'$ and $q'$ do. Finally, we consider whether $t''$ lies strictly to the left of $\overline P_{q,a}$.  If $t''$ lies above $q'$ then it lies strictly to the left of $\overline P_{q',a}$ because that portion of the path lies in $W_{q',a}$.  So suppose $t''$ lies below or at the same $y$-coordinate as $q'$.  Note that $t''$ cannot lie on the path $q, \ldots, q'$ since we assumed that paths $P_{q,a}$ and $P_{t,b}$ do not intersect at a vertex. If $t''$ lies strictly to the right of $P_{q,a}$, then path $P_{q,a}$ must contain a point $p$ strictly inside $W_{z,a}$.  But then $q'$ is not in $W_{p,a}$, a contradiction since the $a$-wedge of a point on an $a$-path must contain all later points of the path.  Thus $t''$ lies strictly to the left of $\overline P_{q,a}$.  
This contradicts the choice of $t'$ as the last vertex of $P_{t,b}$ that is in $W_{q,a}$ and strictly to the left of $\overline P_{q,a}$. 
%
\remove{
We claim that \sout{$q'$ is} the $b$-nearest neighbor of $t'$ {\color{cyan} lies on the subpath $q,\ldots,q'$ of $P_{q,a}$,} and thus that the paths intersect at $q'$.

{\color{cyan}
 First assume that $t'$ is in $W_{q',a}$. If $q'$ lies below $t'$, then the $a$-nearest neighbor $q''$ of $q'$ could be $t'$ (which would yield the required path), or some vertex below or at the same height as $t'$. If $q''$ lies below $t'$, then it would contradict the condition that $q'$ is  the last vertex on $P_{q,a}$ below or at the same height as $t'$, e.g., see Figure~\ref{fig:tddel}{c}. If $q''$ has the same height as $t'$, then it must lie to the right of $t'$, otherwise, it would contradict that $t'$ is to the left of $\overline P_{q,a}$. But having $q''$ with the same height at $t'$ and to the right of $t'$ would contradict that $q'$ is the last vertex on $P_{q,a}$ below or at the same height as $t'$.

 Consequently, $q'$ must be in $W_{t',b}$, e.g., see Figures~\ref{fig:tddel}{d}--{f}. If $q'$ is not the $b$-nearest neighbor of $t'$, then we only need to consider the following two cases. 
 
 (Case 1) Assume that the $b$-nearest neighbor of $t'$ was chosen by breaking ties arbitrarily. If the $b$-nearest neighbor $t''$ of $t'$ lies in $W_{q',a}$ but does not coincide with $q'$, then we can extend $q,\ldots,q'$ to $t''$ (contradicting that $q'$ is the last vertex with height at most that of $t'$). On the other hand, if $t''$ lies to the left of $W_{q',a}$, then  
 $t''$ must lie on the subpath $q,\ldots,q'$, e.g., Figures~\ref{fig:tddel}(d)--(e). Otherwise, we could extend  $t,\ldots,t'$ to $t''$ or some other vertex (in case of a tie), contradicting that $t'$ is the last vertex to the left of  $\overline P_{q,a}$.
 
 (Case 2) The $b$-nearest neighbor of $t'$ was chosen uniquely, e.g., Figure~\ref{fig:tddel}(f).  If $q'$ is not the $b$-nearest neighbor of $t'$, then there must be some better choice $t''$, as follows. If $t''$ lies in $W_{q',a}$, then we can extend  $q,\ldots,q'$ to $t''$, which   contradicts that $q'$ is the last vertex with height at most that of $t'$. Otherwise, if $t''$ lies to the left of $W_{q'a}$, then we can extend  $t,\ldots,t'$ to $t''$ or some other vertex (in case of a tie),   contradicting that $t'$ is the last vertex to the left of  $\overline P_{q,a}$.
}
}
\end{proof} 


\subsection{Layered 3-Sweep Graphs}
\label{sec:full-graph}

In this subsection we define an angle-monotone graph of width $(90^\circ+\alpha)$ for any angle $\alpha$, $0 < \alpha < 45^\circ$, such that $k = \frac{180}{\alpha}$ 
is an integer,
and for any set $S$ of $n$ points. Our graph is defined 
as a $k$-layer 3-sweep graph.

Let $\Delta ABC$ be an acute triangle with $A,B,C$ in clockwise order around the triangle, and with angles 
$\theta_a = 2\alpha$, $\theta_b = \theta_c = 90^\circ - \alpha$.  Orient $\Delta ABC$ so that the vertically upward ray starting at $A$ bisects $\theta_a$.  
 Let $G_1$ be the 3-sweep graph of $S$ with respect to the 3 lines through the sides of $\Delta ABC$.  

We define $G_i$, $2 \le i \le k$ by successive rotations of $\Delta ABC$.  Let $\Delta_i ABC$ be the triangle obtained by rotating $\Delta ABC$ clockwise around $A$ with an angle of $\frac{i-1}{k}360^\circ$, and let $G_i$ be the 3-sweep graph of $S$ with respect to the three lines through the sides of $\Delta_i ABC$. The union of $G_1, \ldots , G_k$ is defined to be the \emph{$k$-layer 3-sweep graph} $H_k$ of $S$ with respect to $\alpha$. 

\begin{theorem}
\label{thm:k3s}
Let $H_k$ be a $k$-layer 3-sweep graph, with $k= \frac{180}{\alpha}$.  Then $H_k$ is an angle-monotone graph of width $(90^\circ+ \alpha)$ and the number of edges in $H_k$ is $O(\frac{n}{\alpha})$. 
\end{theorem} 
\begin{proof}
Let $q$ and $v$ be two points in $S$. Then  $v$ belongs to $W_{q,a}$ in some $G_i$, where ${1\le i\le k}$. By Lemma~\ref{lem:angle-bound}, there exists an angle-monotone path of width $2\alpha+(90^\circ-\alpha)=(90^\circ+\alpha)$ between $q$ and $v$ in  $G_i$, and hence also in $H_k$. By Lemma~\ref{lem:planar}, each $G_i$ is planar. Hence $H_k$ has $O(nk)\in O(\frac{n}{\alpha})$ edges.
\end{proof}

If $2\alpha = 60^\circ$, 
 {then $k=6$.  Because of symmetries, $G_i=G_{i+2}$ so we really only have two 3-sweep graphs, and the resulting graph $H_6$ is the full-$\Theta_6$-graph.}

In the remainder of this section we compare $k$-layer 3-sweep graphs and full-$\Theta_k$ graphs. 
 Figure~\ref{fig:theta8}(a) illustrates the difference. 
 On the one hand, for $k>6$, $H_k$ may have  up to 3 times as many edges as the $\Theta_k$-graph. 
 \changed{However, if $k$ is congruent to 2 mod 4}
 then the sparser graph determined by the union of $G_2,G_4,\ldots,G_k$ has the same properties as $H_k$ \changed{as we now show, using the property that angle-monotonicity is symmetric with respect to the endpoints of the path.} 
\changed{As already noted in the proof of Theorem~\ref{thm:k3s}, for every pair of points $q,v \in S$, $v$ belongs to $W_{q,a}$ in some $G_i$, where ${1\le i\le k}$, and  
then by Lemma~\ref{lem:angle-bound}, there exists an angle-monotone path of width $2\alpha+(90^\circ-\alpha)=(90^\circ+\alpha)$ between $q$ and $v$ in  $G_i$. 
If $i$ is even, this path is included in the union of $G_2,G_4,\ldots,G_k$.  If $i$ is odd, then, because $k$ is even, $q$ belongs to $W_{v,a}$ in $G_{i+(k/2)}$, where the  indices wrap around.  Because $k/2$ is odd, the index $i + (k/2)$ is even, and by Lemma~\ref{lem:angle-bound} there is an angle-monotone path of the required width between $q$ and $v$ in this subgraph. 
}


 On the other hand, every $H_k$ is an angle-monotone graph  of width $(90^\circ+ \alpha)$, but it is not known whether  $\Theta_k$-graphs are angle-monotone with bounded width. For every  $k=4m+4$,  where $m$ is a positive integer, one can construct a $\Theta_k$ graph of width approximately  $(90^\circ+ 2\alpha)$. For example, if $k = 8$, then $2\alpha = 45^\circ$, and $H_8$ is an angle-monotone graph  of width $112.5^\circ$. A $\Theta_8$-graph may have comparatively large  width, e.g., Figure~\ref{fig:theta8}(b) illustrates a $\Theta_8$-graph,  where  any angle-monotone path between $u$ and $v$ has  width  approximately $135^\circ$.

\begin{figure}[htb]
\centering
\includegraphics[width=\textwidth]{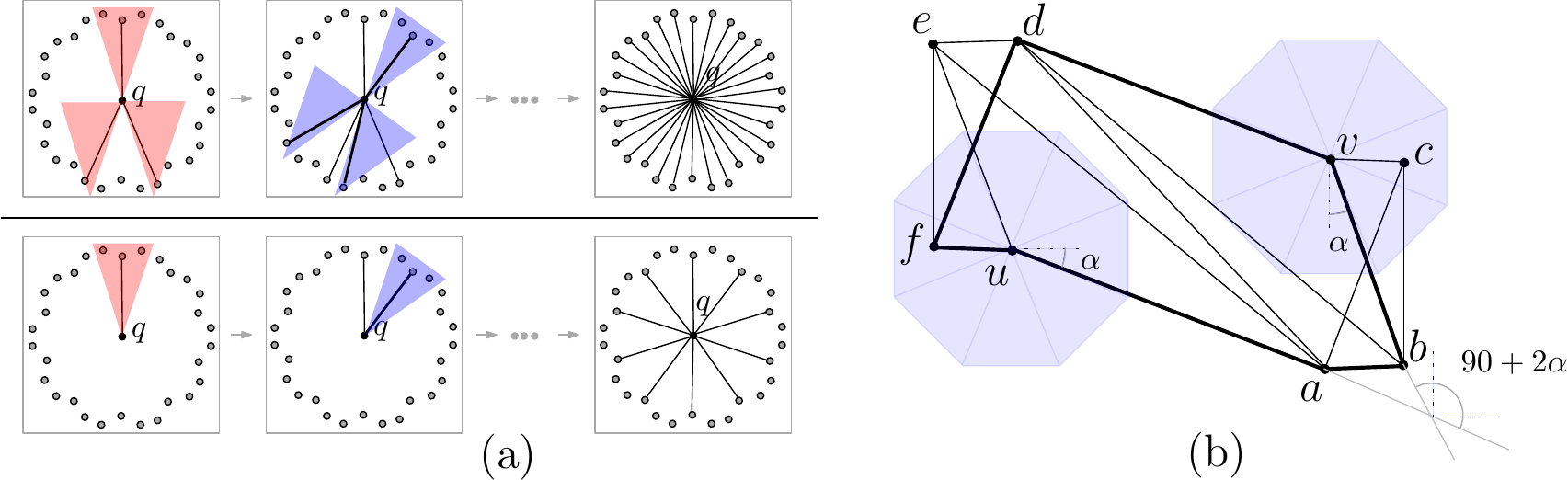}
\caption{(a) Illustration for the neighbors of $q$ in (top) $H_{10}$, and (bottom) $\Theta_{10}$.  (b) An angle-monotone path between $u$ and $v$ of width approximately $(90^\circ+2\alpha) = 135^\circ$ (inspired by an illustration in~\cite{BoseCMRV16}).}
\label{fig:theta8}
\end{figure}

\section{Local Routing in Layered 3-Sweep Graphs}
\label{sec:local-routing}

In this section we give a local routing algorithm for $k$-layer 3-sweep graphs.  Specifically, our routing algorithm is \emph{2-local}, meaning that at each step we assume knowledge of: the coordinates of the current vertex $u$, the coordinates of the target vertex, and the \emph{2-neighborhood} of $u$, which consists of the neighbors of $u$ and their neighbors. In the special case when $k=6$, i.e., for full-$\Theta_6$-graphs, we can restrict ourselves to 1-locality (see Section~\ref{sec:full-theta-routing}).

\begin{theorem}
There is a 2-local routing algorithm
that finds angle-monotone paths of width $90^\circ + \alpha$ in any $k$-layer 3-sweep graph $H_k$, where $\alpha={180^\circ}/{k}$.
The algorithm has routing ratio ${1}/{\cos (45^\circ+\frac{\alpha}{2})}$.
\end{theorem}

 Before giving the algorithm, we explain why we need 2-locality.  Given a start vertex $q$ and a target vertex $t$, we can find, based on the angle of line $qt$, which of the $k$ 3-sweep graphs, say $G_i$ has $t \in W_{q,a}$.  Our routing algorithm will only use edges of $G_i$, so we need a way to tell if an edge of $H_k$ belongs to $G_i$. Consider an edge from current vertex $u$ to some vertex $v$.  From their coordinates, we can decide whether $v$ is in a \emph{positive} wedge of $u$ in $G_i$, i.e., one of $W_{u,a}, W_{u,b}$, or $W_{u,c}$.  If so, then, by checking the other neighbors of $u$, we can detect if $v$ is the unique $a$-, $b$-, or $c$-neighbor of $u$ in that wedge in $G_i$.  Otherwise, $u$ is in a positive wedge of $v$, and, using 2-locality, we can check the neighbors of $v$ to detect if $u$ is the unique $a$-, $b$-, or $c$-neighbor of $v$ in $G_i$.

For the special case of $\alpha=30^\circ$, $H_k$ is the full-$\Theta_6$-graph and our algorithm finds angle-monotone paths of width $120^\circ$ and achieves routing ratio 2.
In this case our algorithm, operating on a single 3-sweep graph, can be viewed as a slight variant of the algorithm of Bose et al.~\cite{BoseTheta6journal} for routing positively in a half-$\Theta_6$-graph. 
 Their algorithm achieves spanning ratio 2 but---as stated--- includes a tie-breaking rule that prevents it from finding angle monotone paths of width $120^\circ$.
 
 %
%
 An example where the tie-breaking rule in that algorithm causes it to find paths of width arbitrarily close to $180^\circ$ is shown in 
Figure~\ref{fig:known}(a)--(c). 
Our contribution is to simplify the statement of the algorithm, generalize to other angles, and give a much simpler proof of correctness using angle-monotonicity.

We briefly mention other approaches to routing.  
The standard \emph{$\Theta$-routing algorithm} forwards the message from the current vertex $v$ either to the destination (if the destination is adjacent to $v$), or to the closest vertex in the cone of $v$ that contains the destination. As illustrated in Figure~\ref{fig:known}(d), the standard $\Theta$-routing algorithm for $\Theta_k$-graphs may also yield paths of large width. We refer the reader to~\cite[Figure~20--22]{BoseCMRV16} for more such examples on full-$\Theta_{4m+4}$ and full-$\Theta_{10}$ graphs.
A recent paper by Bose et al.~\cite{Bose-constraints} gives yet another local routing algorithm for $\Theta_6$-graphs.  This algorithm does find angle monotone paths of width $120^\circ$, but requires knowledge of the source.

\begin{figure}[h]
\centering
\includegraphics[width=.7\textwidth]{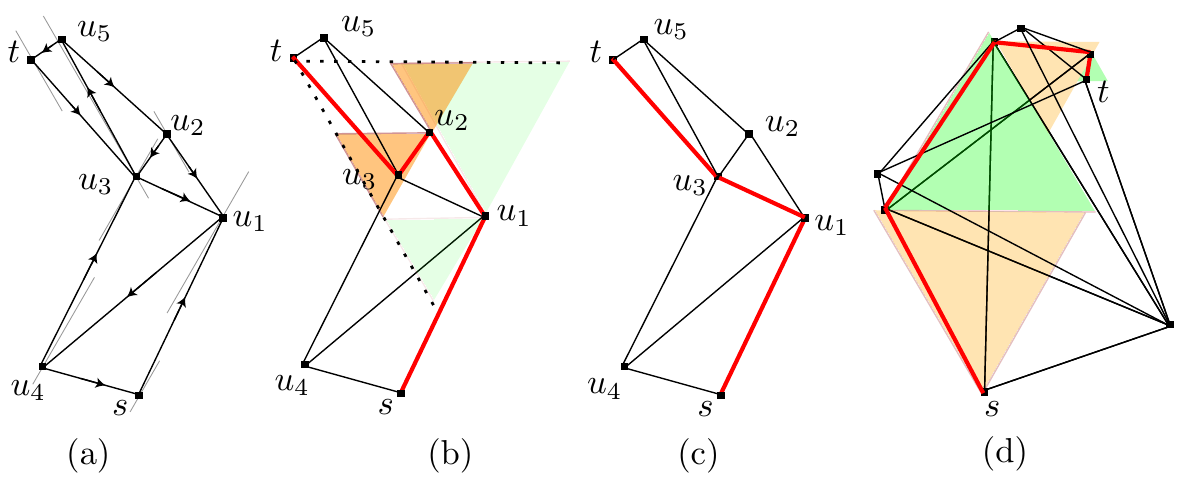}
\caption[Caption without FN]{ (a) A  half-$\Theta_6$-graph $G$ (with arrows indicating the edges into positive cones). 
(b) The routing algorithm  \emph{Algo-half-$\Theta_6$} of~\cite{BoseTheta6journal} on $G$ takes the path $s,u_1,u_2,u_3,t$. The edge from $u_1$ to $u_2$ is chosen following Case B of \emph{Algo-half-$\Theta_6$} (the algorithm favors staying close to the largest empty side)\footnotemark. The edge from $u_2$ to $u_3$ is chosen following Case C of   \emph{Algo-half-$\Theta_6$}. (c) Our algorithm takes the path $s,u_1,u_3,t$. (d)  An example illustrating the standard $\Theta$-routing algorithm on a full-$\Theta_6$-graph.
}
\label{fig:known}
\end{figure}





\paragraph{Algorithm $\mathcal{A}$ (Local Routing)} Let $H_k$ be a  $k$-layer 3-sweep graph  with angles  $\theta_a = 2\alpha, \theta_b = \theta_c = 90^\circ - \alpha$, and let  $q$ and $t$ be two vertices in $H_k$. 
 As discussed above, we can find out which 3-sweep graph, $G_i$, has $t$ in $W_{q,a}$. 
 We will route in $G_i$, using 2-locality to distinguish its edges as discussed above. For ease of description, orient the plane with $W_{q,a}$ pointing upward, centered on the vertical axis, so that edge BC of the reference triangle is horizontal. See Figure~\ref{fig:routing}. 
The general situation is that we have routed (forwarded the message) to some vertex $u$.  Initially $u=q$.   The algorithm stops when $u=t$.

\footnotetext {We note that cases B and D are reversed in the conference version~\cite{Bose:theta6:2012} versus the journal version~\cite{BoseTheta6journal}.}
 \begin{itemize}[noitemsep,topsep=0pt]
 \item While $t$ is an internal point of $W_{u,a}$, forward the message to $u'$, where $u'$ is the 
 $a$-neighbor of $u$ in $W_{u,a}$. See Figure~\ref{fig:routing}(a).  Observe that $u'$ is below or on the horizontal line through $t$.
 
\item At this point, $u$ either belongs to $W_{t,b}$ or  $W_{t,c}$ (possibly lying on the boundary of the wedge).  See Figures~\ref{fig:routing}(b)--(c). If $u$ belongs to $W_{t,b}$, call routine {\bf $\mathcal{A}_L$}, otherwise call routine {\bf $\mathcal{A}_R$}.
\end{itemize}

 

\noindent{\bf Algorithm $\mathcal{A}_L$ (Left Routing).}  Invariant: $u \in W_{t,b}$.  Until $u$ reaches $t$ do the following:
\begin{itemize}[noitemsep,topsep=0pt]
\item Case 1.  Forward the message to 
 the first clockwise neighbor $v$ of $u$ \changed{in $G_i$} such that $v \in W_{t,b}$ and $u \in W_{v,b}$, if such a vertex $v$ exists. 
 See Figure~\ref{fig:routing}(d).
 \item Case 2.
 If no such vertex $v$ exists, then forward the message
 to vertex $u'$, where $u'$ is the 
 \changed{$a$-neighbor of $u$ in $W_{u,a}$}.
 See Figure~\ref{fig:routing}(e). 
\end{itemize} 

\noindent{\bf Algorithm $\mathcal{A}_R$ (Right Routing).}   
Invariant: $u \in W_{t,c}$.  Symmetric to above.
 \medskip

 \begin{figure}[h]
\centering
\includegraphics[width=\textwidth]{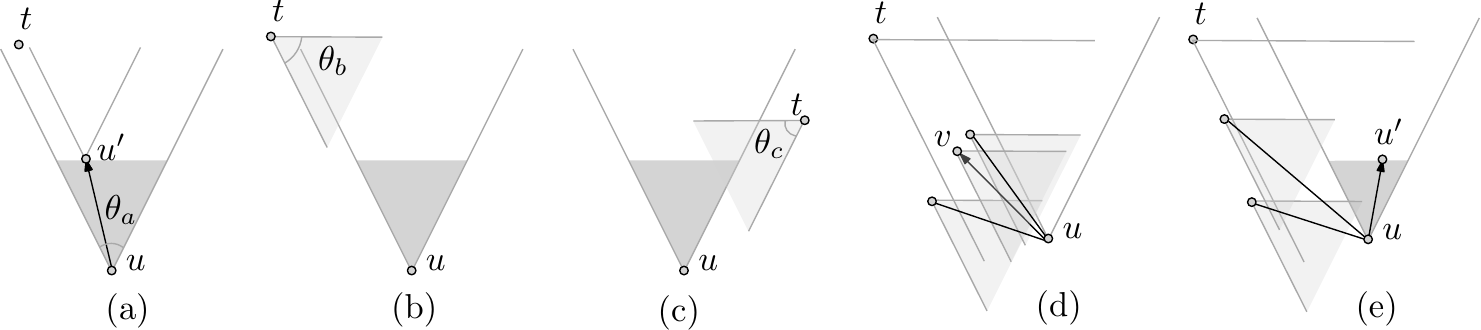}
\caption{Illustration for  algorithm $\mathcal{A}$. 
 }
\label{fig:routing}
\end{figure}
 
 We now prove that $\mathcal{A}$ finds an angle-monotone path of width
 $(90^\circ+\alpha)$ from the source $q$ to the destination $t$.
 Since we execute only one of $\mathcal{A}_L$ or $\mathcal{A}_R$ and they are symmetric, it suffices to consider the case where $\mathcal{A}_L$ is executed. The significant part of the proof is to show that the algorithm finds a path from $q$ to $t$. The fact that the path is angle monotone of width $(90^\circ+\alpha)$ follows immediately. In particular, the initial while loop of algorithm $\mathcal{A}$ uses only $\theta_a$-edges, and algorithm $\mathcal{A}_L$ uses only $\theta_b$- and $\theta_a$-edges.  Thus the path is angle monotone of width $(90^\circ+\alpha)$. 
 Note that the algorithm does not find a path with $\theta_a$-edges appearing before $\theta_b$-edges, as was guaranteed in Lemma~\ref{lem:angle-bound}.
 
In order to show that algorithm $\mathcal{A}$ finds a path from $q$ to $t$ we will show: (1) the invariant $u \in W_{t,b}$ holds for algorithm $\mathcal{A}_L$; (2) some measure improves at each routing step of the algorithm.

First consider the invariant $u \in W_{t,b}$.  
$W_{t,b}$ is bounded by two lines, $\ell$ and $\ell'$, where $\ell$ is the horizontal line through $t$.  To show that $u \in  W_{t,b}$, we must show that $u$ is below, or on, $\ell$, and to the right of, or on, $\ell'$. When we first call $\mathcal{A}_L$, $u$ is to the right of, or on, $\ell'$, and  each step of $\mathcal{A}_L$ preserves this property---see Figures~\ref{fig:routing}(d) and (e). It remains to prove that $u$ is below or on line $\ell$. We will prove the stronger invariant that $P_{t,b}$ goes through or above $u$, i.e.~that $P_{t,b}$
intersects the ray going vertically upward from $u$.

We begin by showing that this is true when we first call $\mathcal{A}_L$.
 If we call $\mathcal{A}_L$ because $q$ is on $\ell'$, then $P_{t,b}$ must pass through or above $q$.  
 The only other way to call $\mathcal{A}_L$ is because  we just completed a step of the while loop of $\mathcal{A}$  where $t$ was internal to $W_{u,a}$ but not internal to $W_{u',a}$, e.g., see Figure~\ref{fig:routing}(a).    By Lemma~\ref{lem:planar}, $P_{t,b}$  cannot cross the edge $(u,u')$. Hence it must pass above or through $u'$.   


Now consider a general step of  $\mathcal{A}_L$.  
We route from $u$  to vertex $w$ which is either vertex $v$ in Case 1 (Figure~\ref{fig:routing}(d)) or vertex $u'$ in Case 2 (Figure~\ref{fig:routing}(e)). Suppose \changed{(for a contradiction)} that the path $P_{t,b}$ does not go through or above $w$.   By induction we know that \remove{the path}$P_{t,b}$ goes through or above $u$. \remove{The path} By Lemma~\ref{lem:planar},  $P_{t,b}$ cannot cross the edge $(u,w)$. 
\changed{(This is where we use the assumption that $(u,w)$ is an edge of $G_i$.)}
Thus  \remove{the path}$P_{t,b}$ must go through $u$ and the other points of edge $(u,w)$ must lie above the path. 
 Let $x$ be the vertex before $u$ on \remove{the path}$P_{t,b}$.  Then  $x \in W_{t,b}$ and $u \in W_{x,b}$.  We now claim that the algorithm should have chosen $x$ rather than $w$.  First note that $x$ is a candidate for vertex $v$ in Case 1 of $\mathcal{A}_L$.  Thus the algorithm would not have moved to Case 2. 
 Next note that $x$ comes before $v$ in clockwise order around $u$, so the algorithm would have chosen $x$ rather than $v$. 

It remains to show that something improves at every step of the algorithm. Let $d_a$ be the distance from $u$ to the horizontal line through $t$. Let $d_b$ be the distance from $t$ to the line determined by the right boundary of $W_{u,a}$. In every iteration of the while loop of $\cal{A}$,  $d_a$ decreases and $d_b$ does not increase. In Case 2 of $\cal{A}_L$, $d_a$ decreases and $d_b$ does not increase. Finally, in Case 1 of $\cal{A}_L$, $d_b$ decreases and $d_a$ does not increase. Thus $d_a + d_b$ strictly improves, and the algorithm must terminate.     
The path found by the algorithm is an angle-monotone path of width $90^\circ+\frac{\theta_a}{2} = (90+\alpha)$.
 

\subsection{1-Local Routing on Full-$\Theta_6$-Graphs}
\label{sec:full-theta-routing}
We observed in Section~\ref{sec:full-graph} that because of symmetries \changed{when $k=6$} we really only have two 3-sweep graphs, which are in fact 
\changed{two half-$\Theta_6$-graphs which together form a full-$\Theta_6$-graph.}
\changed{For this case, we will show that our routing algorithm is 1-local by showing}
how to make the message forwarding decisions based on the 1-neighborhood of the current vertex $u$. Note that the tie-breaking rule that we used to  construct the graph for local routing, i.e., by choosing the most clockwise point, remains the same.

\changed{Suppose without loss of generality that we are routing in $G_1$ and are routing from $u$ to $t$, with $W_{u,a}$ oriented upwards as in the description of 
Algorithm $\mathcal{A}$.}
Recall that while $t$ is an internal point of $W_{u,a}$, Algorithm $\mathcal{A}$ forwards the message to $u'$, where $u'$ is the $a$-neighbor of $u$ in $W_{u,a}$. Since $u$ contains the information about its 1-neighborhood, it is straightforward to make the message forwarding decision.

At this point, we call routine {\bf $\mathcal{A}_L$} or {\bf $\mathcal{A}_R$} depending on whether $u$ belongs to $W_{t,b}$ or  $W_{t,c}$, respectively. By symmetry, it suffices to consider only the left routing  $\mathcal{A}_L$.

Case 1 of  $\mathcal{A}_L$ forwards the message to the first clockwise neighbor $v$ of $u$ \changed{in $G_1$} such that $v \in W_{t,b}$ and $u \in W_{v,b}$, if such a vertex $v$ exists. See Figure~\ref{fig:routing}(d). We now show how to decide the existence of such a vertex $v$ based on the \changed{1-neighborhood} 
of $u$.   

\changed{For any neighbor $q$ of $u$ in $G$, we can easily test if $q \in W_{t,b}$ and $u \in W_{q,b}$.  The issue is whether edge $(u,q)$ lies in $G_1$ or $G_2$.  
Since $q$ lies in a negative cone for $G_1$, edge $(u,q)$ is in $G_1$ if and only if $u$ is a $b$-neighbor of $q$.}
\remove{
We first find a set \changed{$P$} 
of potential candidates from the neighbors of $u$, i.e., 
\changed{$P$ consists of the neighbors of $u$ in $G$ that lie in $W_{t,b}$.}
We now check the vertices of $P$ (in clockwise order around $u$) to see whether 
$u \in W_{q,b}$. The test  whether  $u \in W_{q,b}$ can be done as follows.

{\color{green} J: I am confused about the case 1 of left routing. If we only need $u\in W_{v,b}$, then we can check that standing at $u$. Did we mean $u$ is a $b$-neighbor of $v$? }
}
Let $R_a$ (resp., $R_b$) be the  closed  region determined by the intersection of  $W_{q,b}$ and $W_{u,a}$. 
Let $R$ be the  closed region inside
$W_{q,b}$ bounded by the regions $R_a$ and $R_b$, as shown in gray in Figure~\ref{fig:full-6routing}(a).

\begin{figure}[h]
\centering
\includegraphics[width=.5\textwidth]{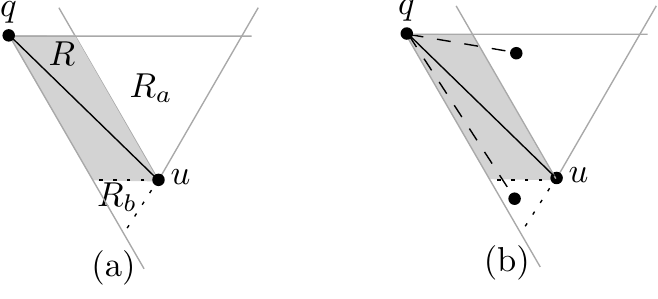}
\caption{Illustration for the routing in a full-$\Theta_6$ graph.}
\label{fig:full-6routing}
\end{figure}

Since the edge $(q,u)$ corresponds to a nearest neighbor in some cone around either $q$ or $u$, the region $R$ cannot contain any vertex except for $u$ and $q$. 
\changed{(This is where we crucially use the fact that $G$ is a full-$\Theta_6$-graph.)}
\changed{Thus $u$ is the $b$-neighbor of $q$ if and only if both $R_a$ and $R_b$ are empty  of any vertex except for $u$ and $q$.  We can test this (even in the general case)}
by checking the $a$- and $b$-neighbors of $u$ in $W_{u,a}$ and $W_{u,b}$, respectively.

Case 2 of  $\mathcal{A}_L$ forwards the message to vertex $u'$, where $u'$ is the nearest neighbor of $u$ in $W_{u,a}$. This case is straightforward since $u$ contains the information about its 1-neighborhood.


\section{Angle-Monotone Graphs with Steiner Points}
\label{app:steiner}
The angle-monotone graphs constructed in  Sections~\ref{sec:subq}--\ref{sec:small-size} can have a width of $90^\circ$ or larger. Spanning graphs of smaller width can be constructed if we allow Steiner points. 
 For example, any spanning planar triangulation with no angle larger than $\beta$ is an angle-monotone graph of width  $\beta$~\cite{Lubiw-O'Rourke}. Since the angles of a triangle sum to $180^\circ$,  the best possible value for $\beta$ is $60^\circ$. We refer the interested reader to~\cite{Bishop16a,BernEG94} for related works on generating meshes with good angle properties. In this section we construct angle-monotone graphs of width $\gamma$, 
 \changed{for any point set $S$ and} 
 any given $\gamma \in (0,90^\circ]$ where $t=(360^\circ/\gamma)$ is an integer. 
 However, the size of the graph depends on some distance parameters of the point set.  


Here we use a pair of non-obtuse triangles 
 $\Delta ABC$ and $\Delta A'B'C'$ to construct the angle-monotone graphs,
 where $\angle BAC$ coincides with $\angle B'A'C'$. 
 Assume that $\theta_a = \theta_{a'} = \theta_{b'} = \theta_{c} = \gamma/2$, as illustrated in Figure~\ref{fig:sp}(a).
 Consider the graph $G_{\it pair}$ obtained by taking the union of 
 the 3-sweep graphs on $S$ with respect to 
 angles $\{\theta_a,\theta_b,\theta_c\}$ and $\{\theta_{a'},\theta_{b'},\theta_{c'}\}$.
 The following lemma is immediate from Lemma~\ref{lem:angle-bound}.
 
\begin{lemma}
\label{lem:angle-bound-2}
 Let $q$ and $t$ be two vertices in $G_{\it pair}$ such that $t$ lies inside $W_{q,\theta_a}$. If $t$ lies to the right of $P_{q,\theta_a}$ (resp., left of $P_{q,\theta_{a'}}$), then there is an angle-monotone path of width $(\theta_a+\theta_c) = \gamma$ (resp., $(\theta_{a'}+\theta_{b'})=\gamma$  between $q$ and $t$.  Figure~\ref{fig:sp}(b) shows the potential location of $t$ in gray.
\end{lemma}

By Lemma~\ref{lem:angle-bound-2}, if $P_{q,\theta_a}$ coincides with  $P_{q,\theta_{a'}}$, then for every point $t \in W_{q,\theta_a}$, we can find an angle-monotone path of width $\gamma$ in $G_{\it pair}$. In the following we show how to insert some additional points (i.e., Steiner points) in $S$ such that we can always find such an angle-monotone path. 
 We refer the reader to Figure~\ref{fig:sp}(c). Assume that $A = A_0$, and let $\ell_{BC}$ be the line through the nearest neighbor $q'$ of $q$ in $W_{q,a}$.  We now construct a sequence of successive triangles $\Delta A_{i-1}B_iA_{i}$ (similar to $\Delta A'B'C'$), where $1\le i\le k$, on the right side of $W_{q,a}$ such that $B_i$ lies on $\ell_{BC}$,
 and then construct Steiner point $s_i$ at $A_i$.

 We choose $k$ to be the smallest integer such that  
 $q'$ does not belong to $W_{s_k,a}$, e.g., see  Figure~\ref{fig:sp}(c), or  the nearest neighbor of $s_k$ in $W_{s_k,a'}$ coincides with $q'$, e.g., see  Figure~\ref{fig:sp}(d). We add the edges $(s_{i-1},s_i)$,  the edge between $q'$ and its corresponding Steiner point, and the edges $(z,s_i)$, where $z\in S$ and the Steiner point $s_i$ is a nearest neighbor of $z$ in $W_{z,b'}$.  Let the resulting graph be $H$, which we will refer to as a \emph{Steiner graph}. Simple trigonometry shows that $A_iC = A_{i-1}C -\frac{A_{i-1}C}{4\cos^2(\gamma/2)}\le \frac{3A_{i-1}C}{4}$. 
 
For every pair of points $u,v\in S$ let  
  $\lambda_{u,v}$ be the smallest distance between
 a pair of non-overlapping parallel lines passing through
 $u$ and $v$ with angle of inclination $90+\frac{\gamma j}{2}$, for some positive integer $j$, e.g., see  Figure~\ref{fig:sp}(e). Define $\lambda$ to be smallest such distance over all $\{u,v\}\in S$. 
 Since at each step the new side length $A_iC$ is at most a fraction of the previous side length $A_{i-1}C$, $k$ is bounded by $O(\log(AC/\lambda))$.

\begin{figure}[h]
\centering
\includegraphics[width=\textwidth]{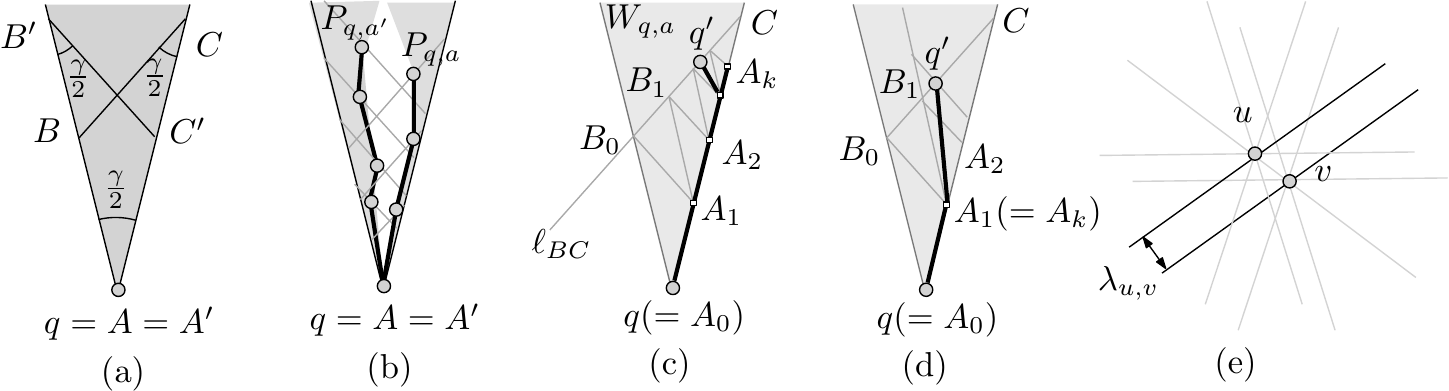}
\caption{(a) $\Delta ABC$ and $\Delta A'B'C'$. (b)--(d) Construction of $H$. (e) $\lambda_{u,v}$. }
\label{fig:sp}
\end{figure}

\begin{lemma}
\label{lem:angle-bound-3}
 Let $q$ and $t$ be two points in $S$ such that $t$ lies inside $W_{q,\theta_a}$. Then there exists an angle-monotone path of width $(\theta_a+\theta_c) = (\theta_{a'}+\theta_{b'})=\gamma$ 
 between $q$ and $t$ in $H$.
\end{lemma}
\begin{proof}
If $t$ lies to the right of $P_{q,a}$, then
  by Lemma~\ref{lem:angle-bound-2}, 
 there exists an angle-monotone path of width $\gamma$ between $q$ and $t$ in $H$.
 
Consider now the case when  $t$ lies to the left of $P_{q,\theta_a}$. 
 Using an analysis similar to the proof of Lemma~\ref{lem:angle-bound} we can observe that $P_{q,a}$ and $P_{t,b'}$ must intersect. If they intersect at a vertex $v$, then the path  $t,\ldots,v,\ldots,q$ is an angle-monotone path of width $\gamma$. Otherwise, let $r$ be the last vertex on $P_{t,b'}$ to the left of $P_{q,a}$, and let $z$ be the last vertex on  $P_{q,a}$  that contains $r$. By construction, $r$ is adjacent to a Steiner point $s$ on the right side of $W_{z,a}$. The path  $t,\ldots,r,s,\ldots,q$ determines the required angle-monotone path of width $\gamma$.
\end{proof}

Let $\mu$ be the largest Euclidean distance determined  by a pair of points in $S$, and let $\mathcal{H}$ be the graph obtained by taking the union of Steiner graphs $H_1,\ldots,H_t$, where $H_i$, $1\le i\le t$, is computed by rotating the plane by $\frac{360^\circ (i-1)}{t})$. Then an analysis similar to the proof of Theorem~\ref{thm:k3s} yields the following result.
  
\begin{lemma}
$\mathcal{H}$ is an angle-monotone graph of width $\gamma$, and the number of edges in $\mathcal{H}$ is  $O(\frac{n}{\gamma}\log \frac{\mu}{\lambda})$. 
\end{lemma}

\section{Open Questions}
\begin{enumerate}
\item (from~\cite{angle-mono}) What is $\gamma_{\rm min}$, the smallest $\gamma$ such that every point set has a planar angle-monotone graph of width $\gamma$? It is known that $90^\circ < \gamma_{\rm min} \le 120^\circ$. 
\item We showed that every set of $n$ points admits an angle-monotone graph of width $90^\circ$ with $o(n^2)$ edges, but can a better bound be proved? $O(n \log n)$ edges?  Even $O(n)$ is not ruled out.  
\item Using Steiner points, we can construct angle-monotone graphs of width $\gamma$, for any given $\gamma>0$, however, the size of the graph depends on some distance parameters of the point set. 
 What is the smallest $\gamma$ such that every point set has an angle-monotone Steiner graph with width $\gamma$  and $o(n^2)$ edges?
\end{enumerate}

\bibliographystyle{plain} 
\bibliography{angle-monotone,incchord}
\end{document}